\documentclass[11pt]{amsart}
\usepackage{amsmath}
\usepackage[dvips]{graphicx}
\usepackage{amsfonts}
\usepackage{amssymb}
\usepackage{amsthm}
\usepackage{mathrsfs}
\usepackage[hyperindex]{hyperref}
\usepackage{amsrefs}

\allowdisplaybreaks[1]
\newtheorem{thm}{Theorem}[section]
\newtheorem{prop}[thm]{Proposition}
\newtheorem{lemma}[thm]{Lemma}
\newtheorem{cor}[thm]{Corollary}

\theoremstyle{remark}

\theoremstyle{remark}
\newtheorem{remark}[thm]{Remark}

\theoremstyle{definition}
\newtheorem{defn}[thm]{Definition}

\theoremstyle{remark}
\newtheorem{ex}[thm]{Example}

\newcommand{\R}{{\mathbb{R}}}

\begin{document}

\title[Apparent horizons with product of spheres topology]{Existence 
of outermost
apparent horizons with product of spheres topology}

\author{Fernando Schwartz}
\address{Mathematics Department, Duke  University}
\email{fernando@math.duke.edu}

\maketitle

\begin{abstract} 

In this paper we find new examples of Riemannian
manifolds with nonspherical apparent horizon, in dimensions 
four and above.
More precisely, for any $n,m\ge 1$, we construct 
asymptotically flat, scalar flat Riemannian manifolds 
containing smooth outermost minimal hypersurfaces
with topology $S^n\times S^{m+1}$.   In the context 
 of general relativity these hypersurfaces correspond to 
 outermost apparent horizons of black holes.
 \end{abstract}

%%%%%%%%%%%%%%%%%%%%%%%%%%%%%%%%%%%%%%%%%%%%%%%%%%%%%%%%%%%%%%
%%%%%%%%%%%%%%%%%%%%%%%%%%%%%%%%%%%%%%%%%%%%%%%%%%%%%%%%%%%%%%
\section{Introduction and main result}
%%%%%%%%%%%%%%%%%%%%%%%%%%%%%%%%%%%%%%%%%%%%%%%%%%%%%%%%%%%%%%
%%%%%%%%%%%%%%%%%%%%%%%%%%%%%%%%%%%%%%%%%%%%%%%%%%%%%%%%%%%%%%

A well known result by Meeks, Simon and Yau \cite{meekssimonyau82}
states that the complement of the region enclosed by the 
outermost minimal surface of a
Riemannian 3-manifold (that is asymptotically flat --AF--
with nonnegative scalar curvature) is diffeomorphic to $\R^3$
minus a finite number of  balls.  This is,  outermost apparent horizons (in three dimensions)
have spherical topology.  Galloway and Schoen's theorem
 \cite{gallowayschoen05}, \cite{galloway06}, which 
 can be thought as a generalization of the above result, 
gives that apparent horizons of $n$-dimensional 
AF manifolds with nonnegative scalar curvature ($n\ge 4$)  
are of positive Yamabe type, i.e. 
admit a metric of positive scalar curvature.  

By shrinking the $S^{m+1}$  factor, it is easily seen that the  product
of spheres $S^n\times S^{m+1}$ has positive Yamabe invariant
for all $n,m\ge 1$.  A natural question 
that arises in view of the above discussion is to determine
whether there exist examples of apparent horizons with
topology $S^n\times S^{m+1}$.  In this paper we address
this question. Our main
result is the following.

 \begin{thm}\label{main} 
For any $n,m\ge1$ there exists an
asymptotically flat, scalar flat $(n+m+2)$-dimensional
 Riemannian manifold $(M,g)$
with outermost apparent horizon which is an outermost smooth 
minimal hypersurface with topology $S^n\times S^{m+1}$.
\end{thm}

Our result also has some relevance in high dimensional relativity,
since it gives suitable initial data (for the Cauchy problem)
with nonspherical apparent horizon that evolves into a black hole.  
  This behavior  is particularly interesting in high dimensions due to 
  string theory (see\cite{peet99}), and also because
there are the only {\it two} known examples 
of nonspherical black holes, both specific to 5 spacetime 
 dimensions.  These are  Emparan and Reall's ``rotating black ring" \cite{emparanreall02}, which is a 5-dimensional spacetime with 
horizon topology  $S^1\times S^2$, and the ``black saturn"
of  \cite{elvangfigueras}, \cite{elvangemparanfigueras}, which
is a black ring rotating around a spherical black hole, and was
obtained by a method that also produces black rings. 
(In \cite{emparan07} there are {\it approximate} spacetime solutions
with black ring topology $S^1\times S^n, n\ge 2$.)\\

 {\it Overview of the proof.}  We construct $(M,g)$ by conformally 
 blowing up an $n$-sphere inside  $\R^{n+m+2}$. We show that this
 manifold has an outermost apparent horizon because we can find barriers.  
 Using
symmetry and geometric measure theory, we 
prove that the horizon is a smooth minimal hypersurface that
(possibly) has a conical singularity at the origin.
By the maximum principle we show that the horizon has two
options. (1) It is smooth everywhere and with topology 
$S^n\times S^{m+1}$, or (2) It is close to having (or has) a 
conical singularity.  Finally, using an ODE analysis 
on the minimal surface equation we prove that a conical singularity 
is an attractor that causes nearby solutions to be 
unbounded. This rules out option (2) since
the horizon is compact.\\

This paper is organized as follows.
In \S \ref{sec1} we introduce notation and give a general argument
for the existence of apparent horizons when barriers are present. In
\S \ref{sec3} we construct $(M,g)$. We also write a simple formula for 
the mean curvature of  symmetric hypersurfaces in it.  We use the formula
to show that $(M,g)$ has barriers. In \S \ref{metric} we prove the horizon 
is a smooth minimal hypersurface possibly with a conical singularity
at the origin.  In \S \ref{topo} we show that the horizon is smooth everywhere and has topology $S^n\times S^{m+1}$.

%%%%%%%%%%%%%%%%%%%%%%%%%%%%%%%%%%%%%%%%%%%%%%%%%%%%%%%%%%%%%%
%%%%%%%%%%%%%%%%%%%%%%%%%%%%%%%%%%%%%%%%%%%%%%%%%%%%%%%%%%%%%%
\section{Preliminaries}\label{sec1}
%%%%%%%%%%%%%%%%%%%%%%%%%%%%%%%%%%%%%%%%%%%%%%%%%%%%%%%%%%%%%%%
%%%%%%%%%%%%%%%%%%%%%%%%%%%%%%%%%%%%%%%%%%%%%%%%%%%%%%%%%%%%%%

Let $(M^n,g)$ denote an $n$-dimensional Riemannian manifold with an
asymptotically flat end $E$. For convenience, compactify
all other ends $E_k$ of $M$ by adding the points $\{\infty_k\}$. 
A marginally-trapped region of $M$ (sometimes called marginally outer trapped) 
is an open set $R\subset M$ that contains the points $\{\infty_k\}$, so that
its boundary $\Sigma=\partial R$ is compact, smooth and has nonpositive mean curvature.  
In our notation: $h_g(\Sigma)\le 0$.  The boundary $\Sigma$ of the marginally-trapped region $R$ is
called a marginally-trapped hypersurface.
Let $\mathscr{R}$ denote the set of all marginally-trapped
regions in $M$.

Following Wald \cite{wald84} we define the apparent horizon of $(M,g)$ as
the boundary of the closure of the union of all marginally-trapped regions.  For our purposes we are only interested in the 
 component of the horizon that contains the end $E$.

\begin{defn} \label{wald}
The apparent horizon of $(M,g)$ is the component $\Sigma^*$  of 
 $\partial (\overline{\cup \mathscr{R}})$ that encloses the end $E$.
\end{defn}

The apparent horizon is unique and outermost since
it encloses all connected marginally-trapped hypersurfaces that
 enclose $E$.  On the other hand, the horizon is not necessarily smooth.
Nevertheless, by symmetry and some other considerations,
we will prove that, in our case, 
the
horizon is a smooth minimal hypersurface.\\

A natural question that arises is to determine geometric conditions
that guarantee the presence of an apparent horizon.  
We say that $(M,g)$ has  
positive ends if it has two or more ends, and each end 
of $M$ can be foliated by smooth compact hypersurfaces 
with positive mean curvature.

A standard argument, for which we only sketch a proof, is the following.

\begin{thm}\label{existz}
A manifold with positive ends
has an apparent horizon.
\end{thm}

\begin{proof}[Sketch of Proof] 
Clearly the set $\mathscr{R}$ is nonempty since the region bounded
by a union of exactly one leaf from each end (except for the noncompactified one) is
marginally trapped.  We claim that no marginally trapped hypersurface can enter the positive
region of the noncompact end, from which it follows that $\overline{\cup\mathscr{R}}\neq M$.
Indeed, suppose that a smooth compact hypersurface
intersects the positive foliation of the noncompact end. At the farthest point with respect
to the foliation, the hypersurface must be tangent to the foliation. By
the maximum principle, the mean curvature of the hypersurface is
positive at that point. Thus it is not marginally trapped.
\end{proof}

\noindent
{\it Note.} In 3 dimensions only spheres can 
be outermost apparent horizons of AF Riemannian manifolds with
nonnegative scalar curvature
by  \cite{meekssimonyau82}.

%%%%%%%%%%%%%%%%%%%%%%%%%%%%%%%%%%%%%%%%%%%%%%%%%%%%%%%%%%%%%%
%%%%%%%%%%%%%%%%%%%%%%%%%%%%%%%%%%%%%%%%%%%%%%%%%%%%%%%%%%%%%%
\section{Existence}\label{sec3}
%%%%%%%%%%%%%%%%%%%%%%%%%%%%%%%%%%%%%%%%%%%%%%%%%%%%%%%%%%%%%%%
%%%%%%%%%%%%%%%%%%%%%%%%%%%%%%%%%%%%%%%%%%%%%%%%%%%%%%%%%%%%%%
 
 We begin this section  by  constructing the manifold $(M,g)$.  
 Our choice of metric $g$ is highly symmetric, as we see below.
 This makes the formula for the mean curvature of 
 symmetric hypersurfaces easy to compute. 
 Near the end of this section we use the mean curvature
 formula for symmetric hypersurfaces to show that $M$ 
 has positive ends. The existence of an 
 apparent horizon in $(M,g)$ follows from Theorem \ref{existz}.
 
Fix $m,n\ge 1$.
 $(M,g)$ is,
 basically,  
 $(n+m+2)$-dimensional  Euclidean space minus an $n$-sphere
 endowed with a conformally related 
 metric that blows up on the $n$-sphere.
  
The construction is as follows.  
  Consider the sphere $S^n=S^n\times\{0\}^{m+1}$ sitting inside the first $(n+1)$-dimensional factor 
of $\R^{n+m+2}$.  Let $G_p$ denote the Green's function 
 for the Laplacian around $p\in \R^{n+m+2}$, i.e. the function 
 $G_p(q)=|p-q|^{-(n+m)}$, where $|p-q|$ is  the Euclidean norm. 
 
 For $\epsilon>0$ (which we later
require to be small) our conformal factor 
 $U$ is the smooth positive function defined on $\R^{n+m+2}\setminus S^n$
by the formula
\begin{align}\label{defu}
U(p)=&1+\epsilon ~G_p*\chi_{S^n}
 = 1+\epsilon\int_{S^n}|p-q|^{-(n+m)}d\mu(q).
\end{align}

\begin{defn} We define
$(M,g)$ to be $\R^{n+m+2}\setminus S^n$
endowed with the conformally flat metric $g=U^{4/(n+m)}\delta_{ij}$.
(For simplicity, $\epsilon$ is removed from the definition of $g$.)
\end{defn}

%The following are basic properties of $g$.

\begin{lemma} \label{asympt} For $(M,g)$ defined as above 
 we have that
\begin{itemize} 
  \item[(i)] $g$ is $SO(n+1)\times SO(m+1)$ invariant, so the
group acts by isometries on $(M,g)$ and fixes both ends.
\item[(ii)] $\Delta_0 U=0$ outside $S^n$, so $(M,g)$  is scalar flat.
\item[(iii)] Let $r(p)=dist(p,S^n)$. For $p$ near $S^n$ we have 
$
 U(p) = 1+ \epsilon r^{-m} + O(r^{1-m}) $ whenever $m>1$. For $m=1, ~ U(p) =1+ \epsilon  r^{-1} + O(\log r) $. 
\item[(iv)] There exists $\alpha=\alpha(n,m)>0$ so that as $|p|\to\infty$, 
$ U(p)=1+\epsilon \alpha |p|^{-(n+m)}+O(|p|^{1-(n+m)}),
$ so spatial infinity  is asymptotically flat.
%\item[(v)] For $r$ as above, $U(r)=1+\epsilon ~ O(\min\{r^{-m},r^{-(m+n)}\}).$ 
\end{itemize}
\end{lemma}

\begin{proof} (i) is direct from the definition of $U$.  For (ii), Green's formula
gives that $U$ is harmonic. The transformation law for the scalar curvature under 
conformal deformations is $R_{g}= -U^{-(n+m+4)/(n+m)}(4(n+m+1)/(n+m)
\Delta_0 U-R_0U)$ which implies $R_g=0$.
 (iii) follows from the expansion $U(p)=1+\epsilon\int_{S^n} (p_1^2+\cdots+ p_{m+1}^2+
\sum_{i=1}^{n+1} (p_{i+m+1}-\xi_i)^2)^{-(n+m)/2}d\mu(\xi)$, which is a particular case
of the calculation in the Appendix of \cite{schoenyau79}. The asymptotic formula in
(iv) holds because of the 
maximum principle together with  $u(p)\to 1$ as $|p|\to \infty$, and 
$u(p)\to+\infty$ as $dist(p,S^n)\to0$.  Standard calculations as those
 in \cite{bartnik86} show that the end is asymptotically flat in this case. 
%(v) follows from (iii) and (iv).
\end{proof}

 The manifold $(M,g)$ is invariant under the action of 
 action of $SO(n+1)\times SO(m+1)$.  
  We are interested in finding a formula for the mean 
 curvature of $SO(n+1)\times SO(m+1)$-invariant 
 hypersurfaces of $(M,g)$ since the apparent
 should also be invariant.  In order to do that, 
 we first find a formula for the mean curvature of
 $SO(n+1)\times SO(m+1)$-invariant hypersurfaces
 of Euclidean space, and then apply the transformation
 law for the mean curvature under conformal deformations.\\

Let $\Sigma\subset (\R^{n+m+2},\delta{ij})$  be an 
$SO(n+1)\times SO(m+1)$-invariant hypersurface.
If we write a vector in $\R^{n+1}\times \R^{m+1}$
as $(x,y)$ respecting this decomposition, then $\Sigma$ has the  form
\begin{equation}\label{sigma}
\Sigma(\gamma)=\{ (x,y) : (|x|,|y|)\in\gamma\},
\end{equation}
where $\gamma$ is some curve in the quadrant  $\R^2_+=\{(a,b)\ge(0,0)\}$.

If  $\gamma=(x(s),y(s))$ is oriented and parametrized by arc length, 
the principal curvatures of $\Sigma(\gamma)$ with respect to an outward 
pointing unit normal are $\kappa, -\dot x/y,\dots, -\dot x/y, \dot y/x,\dots ,\dot y/x$,
where $\kappa$ is the curvature of $\gamma$ with respect to the flat metric. 
 Therefore, the mean curvature of
$\Sigma$ inside Euclidean space is given by
\begin{equation}\label{meancflat}
h_0=\kappa +n\dot y/x-m\dot x/y.
\end{equation}

\begin{ex}\label{example1}
  The coordinate sphere in $\R^{n+m+2}$ is $S_R(0)=
  \Sigma(\gamma),$ where $\gamma=(R\cos s/R,R\sin s/R),~0\le s\le \pi R/2$.  \eqref{meancflat}
  gives  $h_0=(n+m+1)/R$.  
\end{ex}

The transformation law for the mean curvature under 
conformal deformations gives that the mean curvature of a 
hypersurface $N\subset (M,g)$ is 
\begin{equation}
 \label{transfmc}
h=U^{2-c}\left( h_0+
c\partial_\nu\ln U\right),
\end{equation}
where $c=2(n+m+1)/(n+m)$ and $\partial_\nu$ is the Euclidean outward-pointing normal derivative.  

If we 
consider an $SO(n+1)\times SO(m+1)$-invariant hypersurface 
$\Sigma(\gamma)$ like before, and fix
the orientation of $\gamma$ so that the  
outward-pointing normal derivative is given by $\partial_\nu=\dot y \partial_x-\dot x \partial_y$, then the quantity \eqref{transfmc} 
from above is easily
computed using equation \eqref{meancflat}. 
We write this in the following result.

\begin{lemma} Let $\gamma=(x(s),y(s))$ be parametrized by arc length and oriented as above.
Then, the mean curvature of  $\Sigma(\gamma)\subset (M,g)$  is
\begin{equation}\label{meancurv}
 h=  U^{2-c}\left( \kappa+\dot y\left( \frac{n}{x}+c\frac{U_x}{U}
\right) -\dot x\left( \frac{m}{y}+c\frac{U_y}{U}\right) \right),
\end{equation}
where $c=2(n+m+1)/(n+m)$. If $\gamma$ is the graph of $y(x)$,
the mean curvature is given by 
\begin{equation} \label{forgraph}
 h= \pm \frac{U^{2-c}}{\sqrt{1+y'^2}}
 \left( \frac{y''}{1+y'^2}+ y'\left(\frac{n}{x}+c\frac{U_x}{U}\right)
-\left(\frac{m}{y}+ c\frac{U_y}{U} \right)\right).
\end{equation}
\end{lemma}

A discussion on minimal hypersurfaces of the 
Euclidean  metric that are $SO(n+1)\times SO(m+1)$-invariant
 appears in  \cite{ilmanen1998}.  Radially symmetric minimal 
hypersurfaces of Euclidean space have been studied in 
\cite{gilbargtrudinger01}, \cite{gidasninirenberg79} and \cite{struwe96}.
Singular solutions of the minimal surface equation in flat space appear in \cite{bombieridegiorgigiusti69}, where minimal cones in dimension 8 are constructed as 
examples of minimizers that are singular on the origin.  More recent references include
\cite{caffarellihardtsimon84} and \cite{simon97}.

We now use equation \eqref{meancurv} to compute
the mean curvature of  barriers that foliate the ends
of $M$.  The foliations consist of large coordinate spheres on
the asymptotically flat end,  and small 
tubes around the missing $S^n$ on the other
end.  An essential fact is that the tubes around the missing $S^n$
have topology $S^n\times S^{m+1}$ since they are of the
 form $\partial (S^n\times B^{m+2})$.

\begin{thm} \label{theyare}
$M$ has positive ends.
\end{thm}

\begin{proof} 
We begin by foliating a neighborhood of the missing $S^n$ by the tubes
$T_r=\Sigma(\gamma_r),$ where $\gamma_r=(1+r\cos (s/r),r\sin(s/r))$ for $r>0$ small, 
$s\in[0,\pi r]$. 
Using the asymptotics of Lemma \ref{asympt} (iii) and plugging $\gamma_r$
into \eqref{meancurv} gives 
\begin{align*}
 h(T_r)= & -U^{2-c}
 \left(\frac{1}{r}\left[(m-1)-\frac{2m(n+m+1)}{(n+m)(1+{r^{m}}/{\epsilon})}\right]
+\frac{n\cos(s/r)}
{1+r\cos(s/r)}\right)\\
 &+\mbox{lower order terms}\notag,
\end{align*}
which is a positive quantity for $0<r<\epsilon^{1/m}$ and
$\epsilon>0$ small enough. (For this calculation we reverse 
the sign of equation \eqref{meancurv} since the outward-pointing
normal of the tubes $T_r$ points towards the missing sphere.)\\
On the other hand,   the mean curvature of all large enough coordinate spheres $S_R(0)=\Sigma(\gamma_R)$,
where $\gamma_R= (R\cos (s/R),R\sin(s/R))$ for $t\in[0,\pi R/2]$ is positive since $M$ is asymptotically flat.  Indeed,
it follows from example \ref{example1}, equation \eqref{meancurv} and  Lemma \ref{asympt} (iv) that
$h(S_R)=O(1)( (n+m+1)/R-\epsilon \cdot O(R^{-(n+m+2)}))$.
The right-hand side of this equation is 
positive for all $R$ large enough, and continues to be so for all small
values of $\epsilon>0$.  
\end{proof}

\begin{cor}
$M$ has an apparent horizon. 
\end{cor}

\begin{remark} During a conversation with
Hubert Bray we realized that $U$ defined as above makes 
the quantity $h(T_r)$ of
Theorem \ref{theyare} be positive. This observation motivated
much of our work. 
\end{remark}
 
 %%%%%%%%%%%%%%%%%%%%%%%%%%%%%%%%%%%%%%%%%%%%%%%%%%%%%%%%%%%%%%
%%%%%%%%%%%%%%%%%%%%%%%%%%%%%%%%%%%%%%%%%%%%%%%%%%%%%%%%%%%%%%
\section{Regularity}\label{metric}
%%%%%%%%%%%%%%%%%%%%%%%%%%%%%%%%%%%%%%%%%%%%%%%%%%%%%%%%%%%%%%%
%%%%%%%%%%%%%%%%%%%%%%%%%%%%%%%%%%%%%%%%%%%%%%%%%%%%%%%%%%%%%%
 
 In this section we prove that the horizon is a smooth 
minimal hypersurface except possibly at the origin, 
where it may have a conical singularity.  This last
option is ruled out in the next section.  

Recall that
given a curve $\gamma\subset \R^2_+$, the set 
$\Sigma(\gamma)$ is defined as the subset
$\Sigma(\gamma)=\{(|\vec x|,|\vec y|)\in\gamma\}\subset \R^{n+m+2}$.
We  will see that the apparent horizon is given by $\Sigma(\gamma^*)$,
where $\gamma^*$ is a particular curve.  As a matter of fact,
from what we have done
so far we can already say a few things about $\gamma^*$, assuming
this decomposition holds.  For example, the horizon is contained in a compact set,
therefore $\gamma^*$ should  be contained in a bounded region. The horizon encloses the missing sphere,
therefore $\gamma^*$ should separate $\R^2_+$ in two regions, one 
containing all the tubular barriers of Theorem \ref{theyare}.  
This motivates the following definition.
Let $N_\epsilon\subset \R^2_+$ denote a fixed (depends only on $\epsilon$) 
semicircular marginally trapped region that is a neighborhood of $(1,0)$.
(This is possible by the proof of Theorem \ref{theyare}.)

\begin{defn} Let $\gamma:(-L_1,L_2)\to \R^2_+, ~ 0\le L_1,L_2\le \infty$ be a curve parametrized
by arc length.  We say that $\gamma$  
 is minimal if it is smooth, $h(\Sigma(\gamma))=0$, and
  $int(\R^2_+)-\gamma$ has two components, 
one of them being bounded and containing  $N_\epsilon$ from 
above.
\end{defn}

\begin{thm}[Interior regularity]
The apparent horizon is 
$\Sigma^*=\Sigma(\gamma^*)$, where $\gamma^*$  (called the horizon curve) is a minimal curve.
\end{thm}

\begin{proof}
Since $(M,g)$ is $SO(n+1)\times SO(m+1)$-invariant  it follows that
the horizon also is.  Therefore, there is a 
rectifiable curve $\gamma^*$  with $\Sigma^*=\Sigma(\gamma^*)$, and
it follows that $\gamma^*$ is bounded since the horizon is compact. \\
{\it Claim.} $\gamma^*$ is smooth in the interior of $\R^2_+$.\\
Indeed, the apparent horizon outer-minimizes area in the homology class of
a marginally trapped hypersurface. Since the mean curvature of the
marginally trapped hypersurface is nonpositive, a standard argument
shows that the horizon is actually a minimizing
current.  A well known result of geometric measure theory (see e.g. \cite{simon83})
is that the singular set of minimizing currents has codimension 7 or greater. 
On the other hand, singularities of $\gamma^*$ in the interior of 
$\R^2_+$ translate into codimension 2 singularities of 
the horizon $\Sigma^*=\Sigma(\gamma^*)$, which proves the claim.\\
To finalize, recall that a minimizing current is stationary, so
wherever the horizon is smooth it must be a minimal hypersurface.
\end{proof}

 \begin{remark}
 Whenever  $m+n+2\le 7$ the above theorem gives that the
 horizon is a smooth minimal hypersurface everywhere. If
 $m,n\le 6$, the horizon may only have a conical singularity 
 at the origin.   Indeed, since the 
 horizon is minimizing it has codimension 7 
 singularities. If $n+m+2\le 7$,  the singular set of the horizon
 can only be empty. For the other case, note that by symmetry
 the singular set can only have  dimension $1,n$ or $m$. This
 way, if  $n,m\le 6$, the singular set may only be 1-dimensional,
 so it must be a conical singularity at the origin (see Theorem \ref{regu}
 below). 
 \end{remark}

From equation \eqref{meancurv} it follows that the condition
   $h(\Sigma(\gamma))=0$ for minimal
 curves translates into the following equation for the curve
$\gamma(s)=(x(s),y(s)):$
 \begin{equation}\label{mse}
\kappa+\dot y\left( \frac{n}{x}+c\frac{U_x}{U}
\right) -\dot x\left( \frac{m}{y}+c\frac{U_y}{U}\right)=0 ~\mbox{ on }~int(\R^2_+).
\end{equation}
 
Solutions to this equation are graphical almost everywhere with respect to both
axes.

\begin{lemma} \label{graphi}
Except at isolated vertical (resp. horizontal) 
points, a curve that solves \eqref{mse} is locally a graph over the $x$-axis (resp. $y$-axis).
\end{lemma}

\begin{proof}
Suppose $\dot x(s)=0$ for all $s\in(s_0,s_1)$. Writing equation
\eqref{mse} for $s_0<s<s_1$ gives $n/x_0+cU_x/U=0,$  where $x_0=x((s_1+s_2)/2)$.
This is impossible since $U_x/U$ is never constant on vertical segments. A similar
argument works for the horizontal case.
\end{proof}

  We will show that the apparent horizon is a minimal 
 hypersurface smooth everywhere except possibly at the
 origin, where it may have a conical singularity.
 (We remove this last possibility in the next section.)
In order to prove this statement we need the following
 ODE result.  Recall that
$N_\epsilon\subset \R^2_+$ is the marginally trapped neighborhood of $(0,1)$ 
from before.

\begin{prop}\label{foli} 
Smooth solutions of equation \eqref{mse} with an endpoint on the axes
minus $N_\epsilon$ exist uniquely. This is, if $\gamma$ is smooth and solves equation
\eqref{mse} with $\gamma(0)\in (axes-N_\epsilon)$, then $\gamma$ is unique.
Furthermore, $\dot\gamma(0)$ is perpendicular to the axis, unless $\gamma(0)=(0,0)$,
in which case  the angle is $\tan^{-1}\sqrt{m/n}$.  
It follows that there exists a neighborhood $O_\epsilon$ of the axes inside  $\R^2_+-N_\epsilon$ that
is foliated by these unique solutions. (See Figure 1.)
\end{prop}

\begin{proof}
From Lemma \ref{graphi}, solutions of equation \eqref{mse} with
an endpoint on the $y$-axis are graphical over
the $x$-axis for a short while.  In that case, equation \eqref{mse}
becomes 
\begin{equation}\label{mingraph}
 y''/(1+y'^2)+y'(n/x+cU_x/U)-(m/y+cU_y/U)=0.
\end{equation} 
An ODE analysis of this equation gives local existence
of smooth solutions. By inspecting the equation it follows that 
the condition $y'(0)=0$ (or $y'(0)=\sqrt{m/n}$ when $y(0)=(0,0)$) is required.
The same can be done for solutions that intersect the $x$-axis
outside $N_\epsilon$.
\end{proof}

\begin{figure}
\begin{center}
 \includegraphics[height=1.5in]{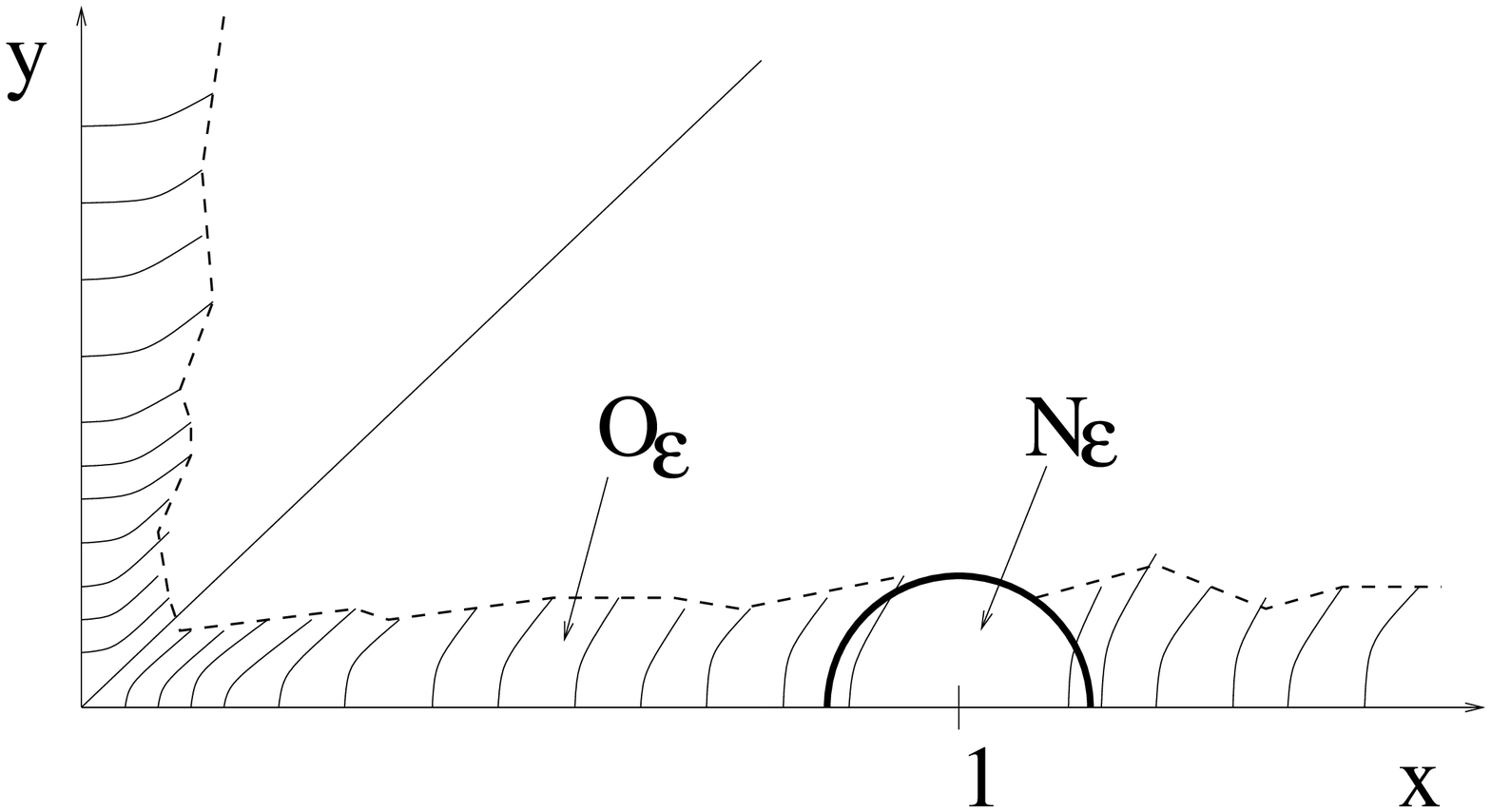} 
\end{center}
Figure 1. The neighborhood $O_\epsilon\subset \R^2_+-N_\epsilon$  is foliated by solutions of \eqref{mse}. 
\end{figure}

Using the ODE proposition we prove the main result of this section.

\begin{thm}[Regularity] \label{regu}
The apparent horizon of $(M,g)$ is a minimal hypersurface 
smooth everywhere except possibly at the origin, where it may 
have a conical singularity.
\end{thm}

\begin{proof} The main idea of the proof is to show that if $\gamma$ is
a minimal curve  that does not match one of the smooth solution curves
that foliate $O_\epsilon$ (from Proposition \ref{foli}) as it approaches the boundary, then 
$\Sigma(\gamma)$ cannot be an area minimizer. We will show this for 
a curve that has an endpoint in the $y$-axis.  The analysis for 
endpoints on the $x$-axis (outside $N_\epsilon$) is analogous.\\
Let $\gamma=(x,y):(-L_1,L_2)\to\R^2_+$ be a minimal curve.  Without loss
of generality we analyze the behavior at $L_2$ only, where we have $\liminf_{s\to L_2} x(s)=0$. \\
{\it Claim 1.} If $\Sigma(\gamma)$ outer-minimizes area in the homology class of a barrier, then
$L_2<\infty$, $y(s),x(s)$ are eventually monotonic as $s\to L_2$ and extend continuously to $(-L_1,L_2]$\\
Indeed, if $\dot y(s_i)=0$ for infinitely many $s_i\to L_2$, then the curve must be tangent to some leaf 
of the foliation of $O_\epsilon$.  By uniqueness of second order ODEs it follows that
$\gamma$ becomes the leaf itself. This gives that $y$ is eventually monotonic and that
$\lim_{s\to L_2}y(s)$ exists and is finite. Now suppose
$x(s)$ has infinitely many local minima as $s\to L_2$, which are arbitrarily
close to the $y$-axis.  By cutting off the curve
at one of these points (that is close enough to the axis) and replacing that segment
by a horizontal line up to the axis, the length of the curve decreases.  
Furthermore, the
cutoff curve continues to enclose the barrier and
the area of $\Sigma(\gamma)$ decreases 
(since the area form is  $x^ny^mU^{2c}ds$). Thus we have found a curve enclosing
the barrier with less area than the horizon --but this is impossible since $\Sigma(\gamma)$ outer-minimizes. 
We deduce that both $x(s), y(s)$ are eventually monotonic. 
Since the curve is bounded and parametrized by arc length, it follows that $L_2$ must be 
finite. This way, the curve extends continuously to $(-L_1,L_2]$ since it is Lipschitz. \\  
From the above argument it follows that $L_1,L_2<\infty$ and $\gamma$ is continuous
on $[0,L]$, where $L=L_1+L_2$. This proves the claim.\\
{\it Claim 2.} $\gamma$ meets the axes perpendicularly.\\
From Lemma \ref{graphi}, $\gamma$ is locally
a graph. In particular, around its endpoint on the $y$-axis, $\gamma$ is the graph of $y(x)$ 
over the interval $[0,\delta)$ for some $\delta>0$.   
Assume initially that $y(0)>0$. An analysis of equation \eqref{mingraph} gives that $y'(0)$ 
exists and is either zero or $\pm $ infinity.  If $y'(0)=\pm \infty$ then $\gamma$ cannot
minimize by a cut-off argument like the one above. This way, it follows that
 $\Sigma(\gamma)$ is smooth so long
as $\gamma$ does not go through the origin.  Nevertheless, in that case, we have that\\
{\it Claim 3.} If $\gamma(0)=(0,0)$ then $\gamma$ meets the origin at an angle $\tan^{-1}\sqrt{m/n}$.\\
To prove this last part we make a blow up argument around the origin. Consider the 
rescaled curve $\gamma_\lambda=(u(s),v(s))=(\lambda x(s/\lambda),\lambda y(s/\lambda))$.  
Multiplying equation \eqref{mse} by $1/\lambda$ and evaluating at 
$s/\lambda$, we get that,
in the limit $\lambda\to\infty$, the rescaled curve solves 
$\kappa+n\dot v/u-m\dot u/v=0$.  Since this equation is invariant 
under rescaling of the graph, its solution must be a line.    
This way $\gamma$ has a tangent line at the
origin. A direct calculation shows that this line has an
angle $\tan^{-1}\sqrt{m/n}$. \\
 It is evident that if $\gamma(0)=(0,0)$ 
  the apparent horizon has a singularity at the origin.  
  By the rescaling argument from above we get that
   the tangent cone of the horizon at the origin is the $SO(n+1)\times 
   SO(m+1)$-invariant minimal cone in $\R^{n+m+2}$. This ends the proof.
\end{proof}

%%%%%%%%%%%%%%%%%%%%%%%%%%%%%%%%%%%%%%%%%%%%%%%%%%%%%%%%%%%%%%
%%%%%%%%%%%%%%%%%%%%%%%%%%%%%%%%%%%%%%%%%%%%%%%%%%%%%%%%%%%%%%
\section{Topology}\label{topo}
%%%%%%%%%%%%%%%%%%%%%%%%%%%%%%%%%%%%%%%%%%%%%%%%%%%%%%%%%%%%%%%
%%%%%%%%%%%%%%%%%%%%%%%%%%%%%%%%%%%%%%%%%%%%%%%%%%%%%%%%%%%%%%

Our goal here is to prove the following result.

\begin{thm} \label{main2}
The apparent horizon is smooth everywhere and has topology $S^n\times S^{m+1}$. 
\end{thm}

Recall that, from before, the horizon 
is given by $\Sigma(\gamma^*)$. Furthermore, the horizon is everywhere smooth, 
unless $\gamma^*$ goes through the origin  --in that case $\Sigma(\gamma^*)$
gets a conical singularity around the origin.\\

{\it Idea of Proof of Theorem \ref{main2}.} 
The main idea of the proof is to prove that the  curve
$\gamma^*$ is (roughly speaking) a small semicircle around the 
point $(1,0)$.  This way,  $\gamma^*$
does not go through the origin, and the horizon 
 does not have a conical singularity.  Furthermore, we get that the topology
of the horizon is automatically $S^n\times S^{m+1}$.\\
In order to show that $\gamma^*$ is ``semicircular'' we 
prove that there are three regions in the plane
that $\gamma^*$ (including its endpoints) must avoid altogether.  Please see Figure 2
for a depiction of them.\\
{\it Region} (1). The marginally trapped region $N_\epsilon$
which is a small enough semicircle around $(1,0)$.  This region is 
already forbidden since
the end is positive. It 
exists by the argument in the proof of Theorem \ref{theyare}.\\
{\it Region} (2). The region $\{y> c'\epsilon ^{1/(n+m)}\}\subset \R^2_+$.
This region is forbidden by the maximum principle which we prove in Subsection \ref{maxx} below. 
(Here $c'>0$ will be a constant independent of $\epsilon$.)\\
{\it Region} (3). A small portion of the axes around the origin. We 
rule out this region in Subsection \ref{last} below. 
This portion is forbidden because there is a conical attractor that 
makes ``almost conical curves"  become complete graphs,
which are unbounded.\\
It follows that once regions (1)--(3) are forbidden, all that $\gamma^*$ can
be is a deformation of a semicircle around $(1,0)$.  This is because $\gamma^*$
is simple, bounded, and  intersects the $x$-axis 
on each side of $(1,0)$.

\begin{figure}
\begin{center}
 \includegraphics[height=1.5in]{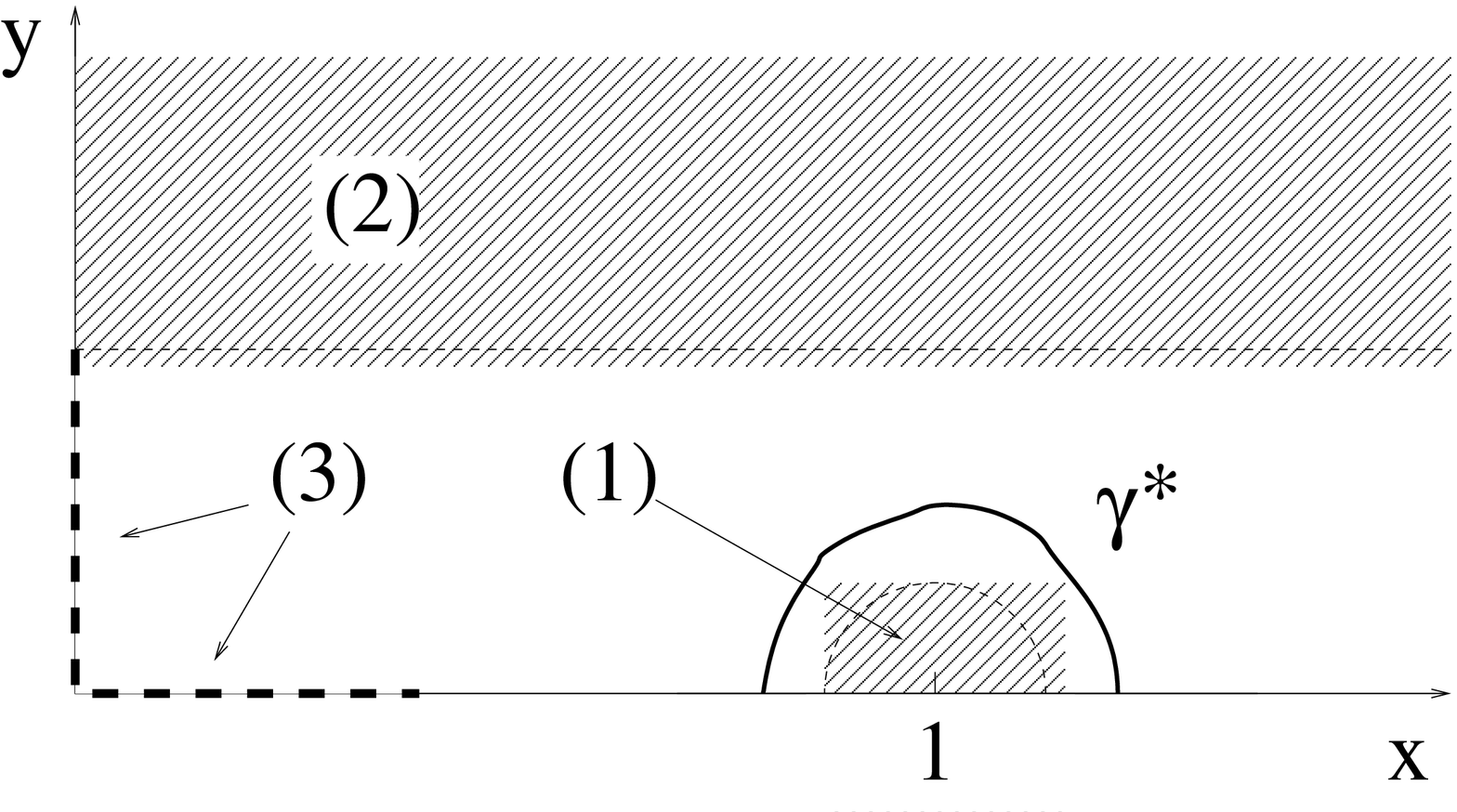} 
\end{center}
Figure 2. The horizon curve $\gamma^*$ and the forbidden regions.

\end{figure}

\subsection{Maximum principle.}\label{maxx}

We begin by finding an estimate for the  derivatives
of $U$ which will be used in the maximum principle.

\begin{lemma}\label{estimate} 
Let $r^2=(x-1)^2+y^2$ as before.  Then
\begin{equation}\label{bound}
 |U_x|  \le \epsilon ~c'r^{-(m+n+1)} ~\mbox{ and }~ 0\le -U_y \le \epsilon ~c'{r^{-(n+m+1)}},
\end{equation}
where $c'>0$ is a constant independent of $\epsilon$.
\end{lemma}

\begin{proof} From the expansion $U(x,y)=1+\epsilon\int_{S^n} (y^2+(x-\xi_1)^2+\xi_2^2
+\cdots +\xi_{n+1}^2)^{-(n+m)/2}d\mu(\xi)$ it follows that
$|U_x|  \le \epsilon ~c'(x+1)r^{-(m+n+2)}$ and that $0\le -U_y\le \epsilon ~c'y r^{-(n+m+2)}$,
which gives the result.
\end{proof}
 
%  The maximum principle is
  
 \begin{thm}[Maximum principle]\label{ruly}
$\gamma^*$ is contained in $\{y\le c'\epsilon ^{1/(n+m)}\}$,
where $c'>0$ is a constant independent of $\epsilon$.
\end{thm} 
 
 \begin{proof} We will show that the maximum height of $\gamma^*$
 is bounded by $c'\epsilon ^{1/(n+m)}$.  Indeed, let $\gamma^*=(x(s),y(s))$.  
 From the proof of Theorem \ref{regu} it follows that $y(s)$ has a global
 maximum at some $s_0>0$ for which $\dot y(s_0)=0$.
We first claim that this global maximum is not located on the $y$-axis.
Indeed, this follows from the estimates of Lemma \ref{taylor1} below, which show that
the curve is concave up around the $y$-axis.  If the curve is oriented
so that $s=0$ at the furthest endpoint on the $x$-axis,  then
at $s_0$ we have $k(s_0)\ge 0$ and $\dot x(s_0)=-1$.  
 Evaluating equation \eqref{mse} at $s_0$  gives $\kappa(s_0) 
 +m/y(s_0)+cU_y/U=0$.  Using the estimate of Lemma \ref{estimate}
 together with $U\ge 1$ and $\kappa\ge 0$ we get 
 $m/y(s_0)\le \epsilon cc'r^{-(n+m+1)}$.  Since $y\le r$, it follows that
 $y(s_0)^{n+m}\le cc' \epsilon/m$. This ends the proof.
  \end{proof}
 
\subsection{The conical attractor.}\label{last}  Here we set
$\delta=\epsilon^{1/(n+m+1)}$ to make some expressions simpler.
 We are interested in proving the
following result.

\begin{thm}\label{unbo} For $\delta>0$  small enough,
any  curve that solves equation \eqref{mse} and has an endpoint on 
$\{0\}\times [0,\delta]$ or on 
 $[0,\delta]\times \{0\}$
 is a complete graph, hence unbounded. 
\end{thm}

 We will give a proof of this theorem for curves that have one endpoint
on $\{0\}\times [0,\delta]$.  The other case is
analogous.

\begin{remark} The depiction of Figure 2 is accurate.  Indeed,
$(\{0\}\times [0,\delta])\cup([0,\delta]\times \{0\})$ 
intersects $\{y>c'\epsilon^{1/(n+m)}\}$ of the maximum principle 
since $\delta=\epsilon^{1/(n+m+1)}$ and  
$ \epsilon^{1/(n+m+1)}>c'\epsilon^{1/(n+m)}$ for $\epsilon>0$ small enough.
\end{remark}

In order to proceed with the proof, we introduce a change of coordinates in equation
\eqref{mingraph} to study its behavior.  Similar coordinates
were used by Ilmanen in \cite{ilmanen1998} to study 
wiggly companions of minimal cones, which satisfy equation \eqref{mingraph}
with $U\equiv 1$. In that case, the change of coordinates makes 
the equation an autonomous two-dimensional system,
 for which stability analysis is simple. 
In our case, our equations are not autonomous so the analysis is 
more involved.\\

Consider the coordinates  
$
t=\log x, W=y/x,~Z=dy/dx.
$
Using the chain rule, equation \eqref{mingraph} 
becomes the first order system 
\begin{equation}\label{phase}
W_t(t,W,Z)=Z-W,~~Z_t(t,W,Z)=n\frac{1+Z^2}{W}(K_1
-K_2WZ),
\end{equation}
where $K_1(t,W)=m/n+cy U_y/(nU)$ and $K_2(t,W)=1+c yU_x/U$.\\

 Let $\gamma$ be a curve that solves equation \eqref{mse} and has an endpoint
on $\{0\}\times [0,\delta]$.  Lemma \ref{graphi} gives that $\gamma$ is the graph
of a function $y(x)$ that solves equation \eqref{mingraph} near the axis. Therefore,
we may use the $(W,Z)$ coordinates and equation \eqref{phase} to obtain the ``phase''
of $\gamma$, i.e. the tuple $(W,Z)$ that represents $\gamma$ 
inside of $WZ$-space (at least for a short time) which is given by the above
change of coordinates.\\

{\it Idea of the proof of Theorem \ref{unbo}.}  Equation \eqref{phase} is ``almost'' autonomous, 
which is the  motivation  for examining  the phase plane
in the  $(W,Z)$ coordinates.  (In the Euclidean case, the analogous equation
becomes autonomous.  See  Figure 3(a), and \cite{ilmanen1998}.)
 We show that there exists a trapping region $\Omega$ in the plane, for which 
the evaluation of the
vector field $(W_t,Z_t)$ on $\partial \Omega$ always points towards
its interior (See Figure 3(b)). Furthermore, we prove that $\Omega$ 
is bounded in the $Z(=dy/dx)$-direction.  
This gives that any curve whose phase enters $\Omega$ 
does not blow up in finite time,
therefore it is a complete  graph and is unbounded.  
In the last part of the proof we find global estimates that  
show that the phase of curves that solve equation \eqref{mse} and have an endpoint on 
$\{0\}\times [0,\delta]$ eventually lies in $\Omega$.  We can expect this behavior
since, in the Euclidean case,  a minimal hypersurface that is 
close to a minimal cone at the origin converges to the cone at infinity.
In particular, it remains graphical and is unbounded.\\

\begin{figure}
\begin{center}
 \includegraphics[height=1.5in]{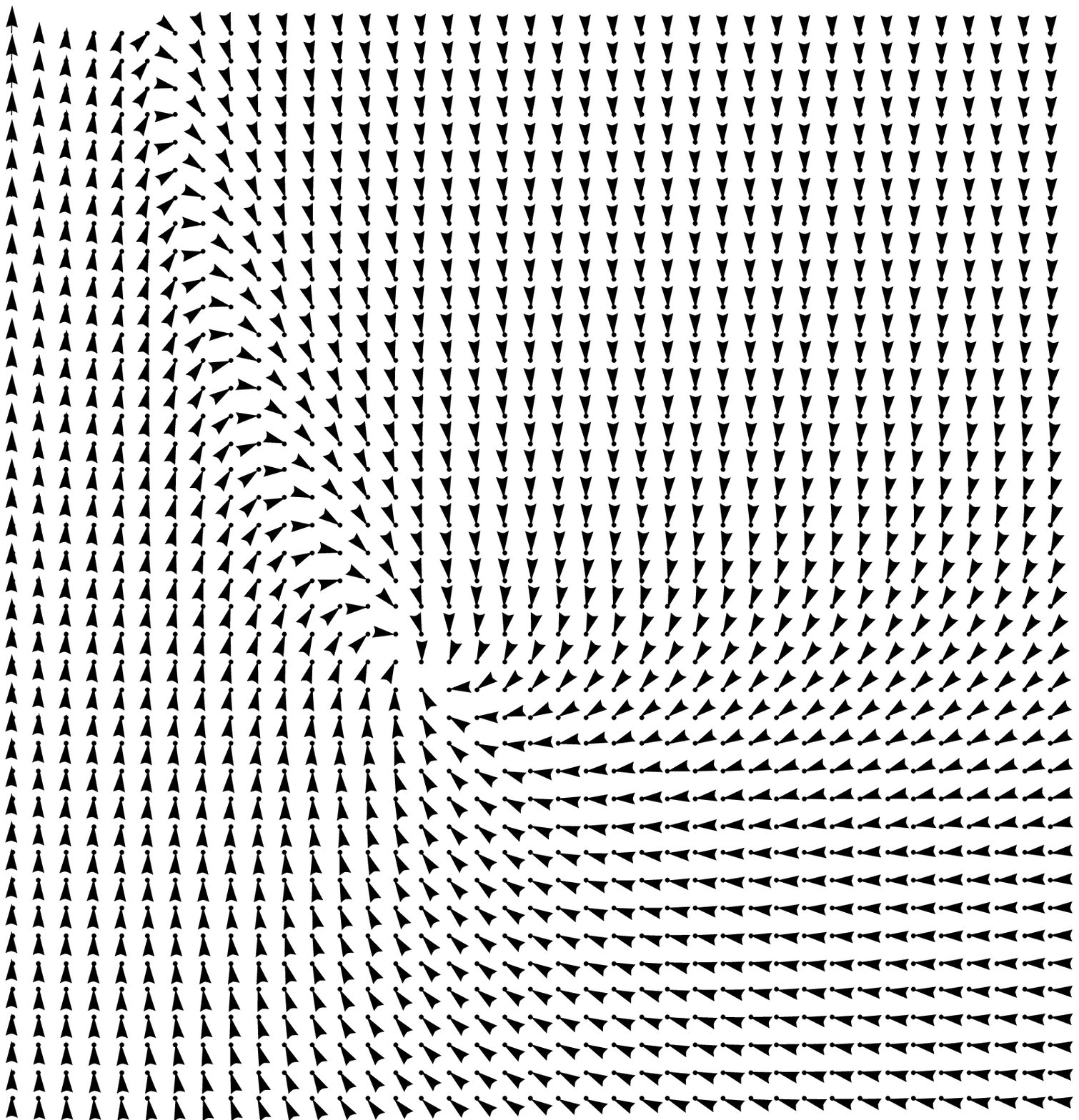} 
\hspace{1cm}
\includegraphics[height=1.5in]{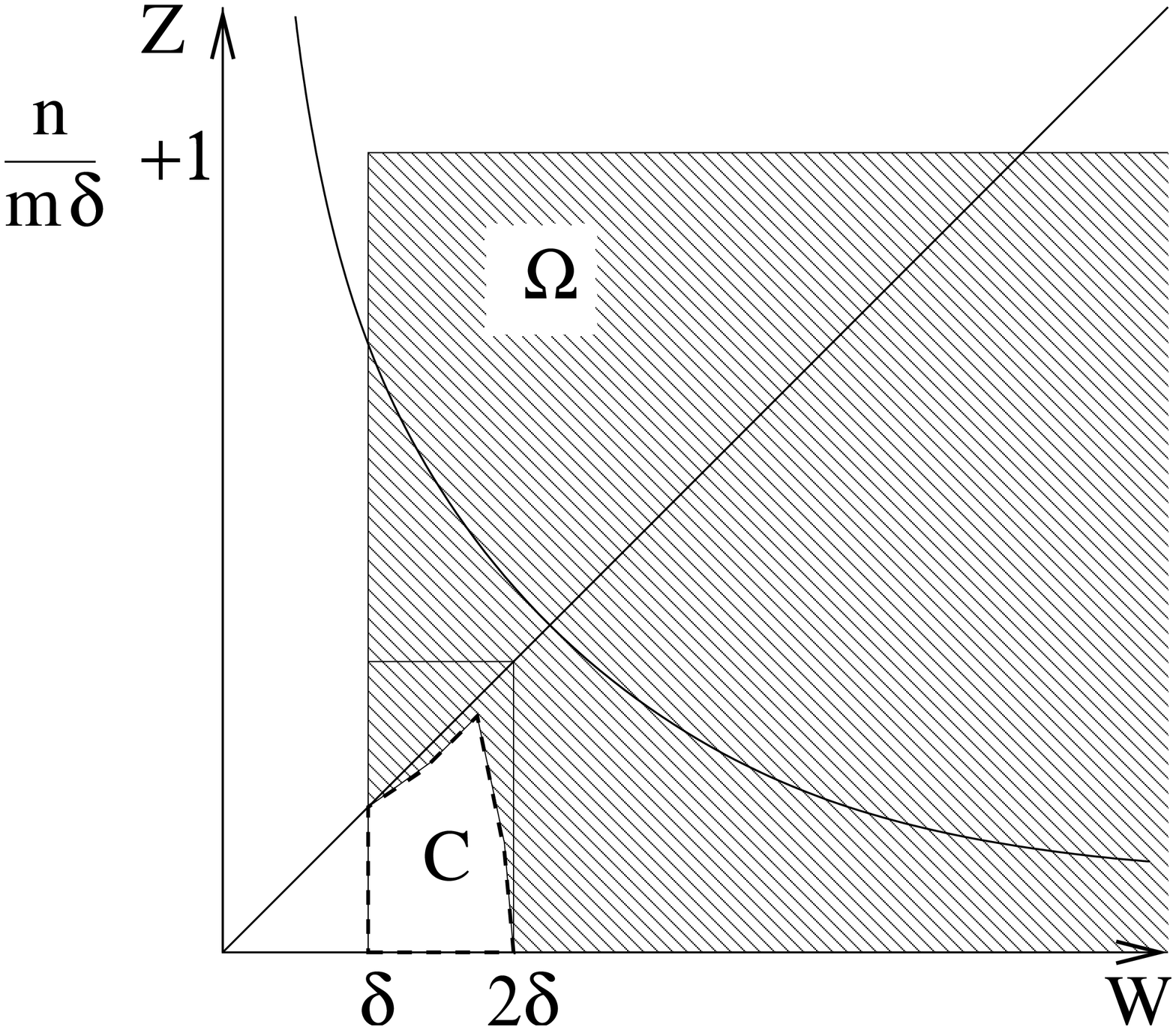}\\
 (a) \hspace{2in} (b)
\end{center}
Figure 3. (a) Direction field $(Z_t,W_t)$ for $U=1$. (b) The trapping region $\Omega$.

\end{figure}

We now introduce the conical attractor, which actually is just an
attracting region. (In the Euclidean case, curves in phase space do
spiral into the attracting fixed point given by the minimal cone.)

 \begin{prop}\label{linear}  Let
$\Omega'=  \{ (\delta,0)\le (W,Z)\le (+\infty, n/(m\delta)+1) \}$ and
$C$ as below.  Then $\Omega=\Omega'-C$
is a trapping region, i.e. the quantity $(W_t,Z_t)$ of equation \eqref{phase} 
evaluated on $\partial \Omega$ always points
towards the inside of $\Omega$.
\end{prop}

The region $C$ is a region 
bounded by two special curves and two lines.  Roughly
speaking, $C$ is a ``worst case scenario region'' for
the vector field, as we see below.    See Figure 3(b) for a 
depiction of all of them.

We begin by bounding the terms $K_1,K_2$ from above.  We will later use
these bounds to construct the worst case scenario  region $C$.

\begin{lemma} Let $K_1,K_2$ be as in the system \eqref{phase}. On the region $\{(t,W)~:~W\ge \delta,~t\in \R\}$ we have
the uniform bounds
\begin{align*}\label{bound2}
 \frac{m}{n}-c'\epsilon^\alpha& \le K_1(t,W)\le \frac{m}{n},~\mbox{ and }~1
 -c'\epsilon^\alpha\le 
K_2(t,W)\le 1+c'\epsilon^\alpha,
\end{align*}
where $\alpha=1-(m+n)/(m+n+1)>0$.
\end{lemma}

\begin{proof}
$W\ge \delta$ is just $y\ge\delta x$, so it must be $r\ge\delta$ in that region. The rest
follows from  applying the estimate \eqref{bound}.
\end{proof}

A consequence of this estimates are global  bounds for the vector field $(W_t,Z_t)$
in a particular region of the $WZ$-plane.

\begin{cor}\label{ll} Let $R=[\delta,2\delta]\times [0,2\delta]$ and $(W_t,Z_t)$ from \eqref{phase}. Then
\begin{align*}
\inf_{\R\times R} W_t\ge -2\delta, ~\mbox{ and }~\inf_{\R\times R} Z_t\ge \frac{m-nc'\epsilon^\alpha}{2\delta}-\delta(1+c'\epsilon^\alpha)(2n+8\delta^2), 
\end{align*}
where $\alpha=1-(m+n)/(m+n+1)>0$.
\end{cor}

The set $C$ is constructed as follows. Let $\sigma$ denote
an integral curve of $(W_t,Z_t)$ in the $WZ$-plane that starts at the 
bottom right corner of $R$, i.e. at the point $p=(2\delta,0)$. (Note: such a curve exists by
general ODE results, but it is not necessarily unique since the system \eqref{phase} 
is not autonomous.)  Let $R_1(\sigma)\subset R$
denote the region inside $R$ that is bounded on top by the diagonal and bounded 
on the right 
by the curve $\sigma$.  We define $C_1$ to be the intersection of 
all the possible $R_1(\sigma)$, where $\sigma$ is an integral curve that starts at $p$. 
(This is, $C_1$ is a worst case scenario region for integral  curves that start at $p$.) 
$C_1$ is not empty.  Indeed, a direct calculation using
 Corollary \ref{ll} shows that

\begin{lemma} Let $L$ be the line $Z(W)=-W(m/4\delta^2)+m/2\delta^2$. Then
$C_1$ includes the region $R_1(L)$, and $R_1(L)\supset (\{\delta\}\times
 [0,\delta]\cup [\delta,2\delta]\times\{0\})$.
\end{lemma}

Now let $\sigma$ denote an integral curve that starts at $q=(\delta,\delta)$, and let
$R_2(\sigma)$ denote the region inside $C_1$ bounded on top by the curve $\sigma$.
Note that any such $\sigma$ is increasing (as a graph on the $W$-axis)
since $W_t$ is positive inside $R$.
We define $C$ to be the intersection of all 
the possible $R_2(\sigma)$, where $\sigma$ is an integral curve that starts
at $q$. (This is, add to $C_1$ a worst case scenario for integral curves that start at $q$.)
$C$ is not empty since $\sigma$ is increasing.  Indeed, 
we are, at worst, capping off the top of $C_1$  by a horizontal line through $q$.
See Figure 3(b).

\begin{proof}[Proof of Proposition \ref{linear}]
The boundary of $\Omega$ consists of five parts as depicted in Figure 3(b).  
The roof is given by
a segment of the line $\{Z=n/(m\delta)+1\}$ that lies well inside the region  
$Z=K_1/(K_2W)$.  This gives that $Z_t$ points downward on the top.  The bottom
boundary is part of the positive $W$-axis, and $Z_t$ is positive there.  
The left-side boundary consis of three segments.  One is the segment of the
line $\{W=\delta\}$ that lies above the diagonal, for which $W_t$ is nonnegative.
What is left is the top and right of the boundary of $C$.  From 
the construction it follows that $(W_t,Z_t)$ points inwards there
as well.
\end{proof}

We now show that the phase of any  curve that solves equation \eqref{mingraph} and has an 
endpoint on $\{0\}\times [0,\delta]$
eventually lies in $\Omega$.  In order to do that, we compute a second order
approximation of these solutions. We check that after a
small time they are close to the conical solution. (See Figure 1 for a
snapshot of this phenomena.)

\begin{lemma}\label{taylor1}
A smooth solution of  equation \eqref{mingraph} with  $y(0)>0$ satisfies 
\begin{align*}
 y''(0)= &\frac{m}{(n+1)}\frac{1}{y(0)} + c\frac{U_y/U-U_x/U}{n+1}    
 = \frac{m}{(n+1)}\frac{1}{y(0)}+\epsilon ~O(r^{-(m+n+1)}).
\end{align*}
 \end{lemma}

\begin{proof}
 We know from the proof of Theorem \ref{regu} that $y'(0)=0$. Multiplying \eqref{mingraph} by $y$,
taking limit $x\to 0$ and using L'Hopital's rule in the term $y'/x$ gives
$(n+1)y''(0)y(0)-m+ y(0)c(U_x/U-U_y/U)=0$.  Together with the gradient estimate
of Lemma \ref{estimate} and the fact that $U\ge 1$ this gives the result.
\end{proof}

\begin{cor} \label{taylor2}
Suppose $y$ solves equation \eqref{mingraph}
with $y(0)=\delta$ for $\delta>0$ small enough.  Then for $0\le x\le \delta$, we have that
\begin{align*}
 y(x)&= \delta+\frac{m}{2(n+1)\delta}x^2+O(\epsilon), ~\mbox{ and }~y'(x)=\frac{mx}{(n+1)\delta}+O(\epsilon).
%\mbox{ and } ~y''(x)&=\frac{m}{(n+1)\delta}+O(\epsilon).
\end{align*}
\end{cor}

\begin{proof} Follows from
Lemma \ref{taylor1} and Taylor's theorem.
\end{proof}

\begin{prop} \label{tom} For $\delta\ge0$ small enough, the phase of
any solution $y$  of equation \eqref{mingraph} with $y(0)= \delta$ 
eventually enters the region $\Omega$.
\end{prop}

\begin{proof} We first prove this for a curve with $y(0)=0$.  Indeed, in that case
$y(x)=\sqrt{m/n}x+O(\epsilon)$ and $y'(x)=\sqrt{m/n}+O(\epsilon)$, from before.  
This way, for small $x>0$ we get $W=y/x\approx \sqrt{m/n}$ and $Z=y'\approx \sqrt{m/n}$, 
so $(W,Z)$ lies inside $\Omega$.\\
Whenever $y(0)=\delta>0$ is small, using the expansions of Corollary \ref{taylor2} 
evaluated at $x=\delta$
give $W=y(\delta)/\delta \approx 1+m/2(n+1)$ and $Z=y'(\delta)\approx m/(n+1)$, which also gives a point
that lies within $\Omega$.
\end{proof}

\begin{proof}[Proof of Theorem \ref{unbo}]
Follows directly from the above proposition together with 
Proposition \ref{linear}.  This ends the proof of Theorem \ref{main} as well.
\end{proof}

{\bf Acknowledgments.} I would like to thank Hubert Bray for his encouragement and 
many useful discussions.  I also thank
 Bill Allard, Greg Galloway, Bob Wald and Brian White for 
their useful conversations.

%%%%%%%%%%%%%%%%%%%%%%%%%%%%%%%%%%%%%%%
%%%%%%%%%%%%%%%%%%%%%%%%%%%%%%%%%%%%%%%
%%%%%%%%%%%%%%%%%%%%%%%%%%%%%%%%%%%%%%%
%%%%%%%%%%%%%%%%%%%%%%%%%%%%%%%%%%%%%%%
%%%%%%%%%%%%%%%%%%%%%%%%%%%%%%%%%%%%%%%
%%%%%%%%%%%%%%%%%%%%%%%%%%%%%%%%%%%%%%%
%%%%%%%%%%%%%%%%%%%%%%%%%%%%%%%%%%%%%%%
%%%%%%%%%%%%%%%%%%%%%%%%%%%%%%%%%%%%%%%
%%%%%%%%%%%%%%%%%%%%%%%%%%%%%%%%%%%%%%%
%%%%%%%%%%%%%%%%%%%%%%%%%%%%%%%%%%%%%%%
%%%%%%%%%%%%%%%%%%%%%%%%%%%%%%%%%%%%%%%
%%%%%%%%%%%%%%%%%%%%%%%%%%%%%%%%%%%%%%%
%%%%%%%%%%%%%%%%%%%%%%%%%%%%%%%%%%%%%%%
%%%%%%%%%%%%%%%%%%%%%%%%%%%%%%%%%%%%%%%
%%%%%%%%%%%%%%%%%%%%%%%%%%%%%%%%%%%%%%%
%%%%%%%%%%%%%%%%%%%%%%%%%%%%%%%%%%%%%%%
%%%%%%%%%%%%%%%%%%%%%%%%%%%%%%%%%%%%%%%
\end{document}